\documentclass[a4paper,UKenglish]{lipics-v2018}

\usepackage{booktabs} 

\usepackage[ruled]{algorithm2e} 
\usepackage{xargs}
\usepackage{marginnote}
\usepackage[pdftex,dvipsnames]{xcolor}  
\usepackage[colorinlistoftodos,prependcaption,textsize=tiny]{todonotes}
\usepackage[inline]{enumitem}

\newcommandx{\note}[2][1=]{\todo[linecolor=blue,backgroundcolor=blue!25,bordercolor=blue,#1]{#2}}

\theoremstyle{definition}

\newcommand{\xp}[1]{\ensuremath{f^ {-1}({#1})}}

\newcommand{\mononm}{core minimal\xspace}
\newcommand{\semantic}[1]{\ensuremath{\mathtt{sem}\mbox{-}{#1}}}
\newcommand{\sghw}{\ensuremath{\semantic{ghw}}}
\newcommand{\sfhw}{\ensuremath{\semantic{fhw}}}
\newcommand{\sadw}{\ensuremath{\semantic{adw}}}
\newcommand{\ssubw}{\ensuremath{\semantic{subw}}}

\title{Semantic Width of Conjunctive Queries and Constraint Satisfaction Problems}

\author{Georg Gottlob}{University of Oxford, UK \& TU Wien, Austria}{georg.gottlob@cs.ox.ac.uk \& gottlob@dbai.tuwien.ac.at}{}{}

\author{Matthias Lanzinger}{TU Wien, Austria}{mlanzing@dbai.tuwien.ac.at}{}{}

\author{Reinhard Pichler}{TU Wien, Austria}{pichler@dbai.tuwien.ac.at}{}{}

\authorrunning{G. Gottlob, M. Lanzinger and R. Pichler} 

\Copyright{Georg Gottlob, Matthias Lanzinger and Reinhard Pichler}

\subjclass{H.2.4 Systems - Query processing}
\keywords{Conjunctive queries; Constraint satisfaction problems, hypergraphs; semantic optimization; semantic width;
fractional hypertree width; submodular width; decompositions}

\nolinenumbers 

\begin{document}
\hideLIPIcs

\maketitle

\begin{abstract}
Answering Conjunctive Queries (CQs) and solving Constraint Satisfaction Problems (CSPs) are arguably 
among the most fundamental tasks in Computer Science. They are classical NP-complete problems. 
Consequently, the search for tractable fragments of these problems has received a lot of 
research interest over the decades. This research 
has traditionally progressed along three orthogonal threads.
  \begin{enumerate*}[label=\itshape\alph*\upshape)]
  \item Reformulating queries into simpler, equivalent, queries
    (semantic optimization)
    
  \item Bounding answer sizes based on
    structural properties of the query
  \item Decomposing the query in such a way that global consistency follows from local consistency.
\end{enumerate*}

Much progress has been made by various works that connect two of these threads. Bounded answer sizes and decompositions have been shown to be tightly connected through the important notions of fractional hypertree width and, more recently, submodular width. Recent papers by Barcel\'{o} et al. study decompositions up to generalized hypertree width under semantic optimization. In this work, we connect all three of these threads by introducing a general notion of \emph{semantic width} and investigating semantic versions of fractional hypertree width, adaptive width, submodular width and the fractional cover number.
\end{abstract}

\section{Introduction}
\label{sec:intro}

Answering conjunctive queries (CQs) is one of the central themes of
database theory. As the problem is known to be NP-complete in general,
identifying tractable classes of CQs has been the focus of much
research~(see e.g., \cite{gottlob2016hypertree} and many references therein). 
The study of bounded decompositions and answer
sizes in particular has produced many remarkable results and has
significantly advanced our understanding of what makes CQs \emph{hard}
in general.

A largely separate line of research has followed the question of
minimizing the query itself. Given a query it seems enticing to only
want the minimal (in the number of atoms) version that produces the
same result, i.e., a semantically equivalent query. This is referred
to as semantic optimization. In a classic result, Chandra and
Merlin~\cite{chandra1977optimal} show that queries are (semantically)
equivalent iff they are homomorphically equivalent. The minimal
equivalent CQ is called the core of a query.  However, while semantic
optimization simplifies the query in a sense, even finding the minimal
equivalent query still provides little information on the complexity
of its execution.

The equivalence classes induced by semantic equivalence can be viewed
as the collection of all possible ways to formulate the same question
to the database (in the language of CQs). It is then natural to ask if
some formulations can be answered more efficiently than others and how
to identify and derive these efficiently answerable formulations. Work
by Dalmau et al.~\cite{dalmau2002constraint} and Barcel\'{o} et
al.~\cite{barcelo2016semantic, barcelo2017semantic} investigates the
treewidth, acyclicity and generalized hypertree width (ghw) under
semantic equivalence. Specifically in the case of acyclicity the term
\emph{semantic acyclicity} is used. This theme motivates the
introduction of \emph{semantic width} (for any notion of width) as a
measure for the complexity of the underlying question to the database
rather than the complexity of a specific formulation.

\begin{definition}
  \label{def:dodgysemwidth}
  The semantic width of a
  conjunctive query $q$ is the minimal width over all conjunctive queries that are
  equivalent to $q$.
\end{definition}

In this sense, the main goal of this work is the investigation of
\emph{semantic fractional hypertree width}, \emph{semantic adaptive
  width} and \emph{semantic submodular
  width}. Section~\ref{sec:prelim} provides preliminaries and formal
definitions of the central concepts. In Section~\ref{sec:semprop} we introduce a formal version of 
Definition~\ref{def:dodgysemwidth} and our central machinery. We show that
the \emph{semantic fractional cover number} (analogous to semantic width) is determined by the fractional cover number of the core in Section~\ref{sec:homgraph}. The same result is then also derived, for all the semantic widths stated above, in Section~\ref{sec:swidth}.

\section{Preliminaries}
\label{sec:prelim}

The initial definitions here are adapted from~\cite{gottlob2002hypertree}. 
A relation schema $R$ consists of a name $r$ and 
an ordered list of
attributes. An \emph{attribute} $a$ has an associated countable
\emph{domain} $dom(a)$. A \emph{relation instance} of a schema with
attributes $(a_1, \dots, a_k)$ is a finite subset of
$dom(a_1) \times \dots \times dom(a_k)$. The elements of relation instances are called \emph{tuples}. A \emph{database schema} is a
finite set of relation schemas. A \emph{database (instance)}
\textbf{D} over a database schema $\{R_1, \dots, R_m\}$ consists of
relation instances for every schema $R_1,\dots,R_m$. The
\emph{universe} of a database \textbf{D} is the set of all values occurring for attributes of the
relation instances of \textbf{D}.

A (rule based) \emph{conjunctive query} $q$ on a database schema
$DS=\{R_1, \dots, R_m\}$ is a rule of the form
$q\colon ans(\overline{x}) \leftarrow r_1(\overline{x_1}) \land \cdots
\land r_n(\overline{x_n})$ where $r_1, \dots,r_n$ are relation names
of $DS$, $ans$ is a relation name not in $DS$ and
$\overline{x},\overline{x_1},\dots,\overline{x_n}$ are ordered lists
of terms matching the number of attributes of the respective relation
schema. We refer to the atoms in the body as $atoms(q)$ and the
variables occuring in an atom $R$ as $var(R)$.

The \emph{answer} of $Q$ on a database \textbf{D} with universe $U$ is
a relation $ans$ with attributes $\overline{x}$. The tuples of $ans$
are all tuples $ans(\overline{x})\sigma$ such that:
$\sigma \colon  var(Q) \to U$ is a substitution and
$r_i(\overline{x_i}) \sigma$ is in a relation instance of \textbf{D}
for $1\leq i \leq n$. A substitution $\sigma$ is applied to an atom
$A$ by replacing each variable $X$ in $A$ with $\sigma(X)$.

For conjunctive queries $q_1,q_2$, a \emph{homomorphism} from $q_1$ to $q_2$ is a mapping $f\colon vars(q_1) \to vars(q_2) \cup constants$ such that
\begin{enumerate}[topsep=0pt]
\item For every variable $x$ in the head of $q_1$, $f(x) = x$
\item For every atom $R(x_1, \dots, x_k) \in atoms(q_1)$ there exists an atom $R(f(x_1), \dots, f(x_k)) \in atoms(q_2)$.
\end{enumerate}
If there exists a homomorphism in both directions we say $q_1$ and $q_2$ are \emph{homomorphically equivalent}, we write $q_1 \simeq q_2$.

We call CQs equivalent if they have the same answer over any database
instance. It is known that two queries are equivalent iff they are
homomorphically equivalent~\cite{chandra1977optimal}. For a CQ $q$, a
minimal (in the number of atoms) equivalent CQ is called a core. All
cores of $q$ are isomorphic and it is therefore common to refer to
\emph{the} core of $q$ (we sometimes write $Core(q)$). Another
observation by Chandra and Merlin is important in our context: The
core can always be obtained by an endomorphism on the query or
equivalently by deleting atoms.

A \emph{hypergraph} $H = (V(H), E(H))$ is a pair, where $V(H)$ is a the set of vertices and the set of \emph{hyperedges} $E(H)$ is a set of subsets of $V(H)$. In this work we assume hypergraphs are finite, and all vertices are contained in some hyperedge (there are no isolated vertices).  We write $I_v$ for the set of all
incident edges of a vertex $v$.  A \emph{homomorphism} $G \to H$ for
hypergraphs is a mapping $f \colon V(G) \to V(H)$ s.t. if $e \in E(G)$,
then $\{ f(v) \mid v \in e \} \in E(H)$. Function application is
extended to hyperedges and sets of hyperedges in the usual, element-wise, fashion:  for instance, 
for $e \in E(G)$, we write $f(e)$ to denote $\{ f(v) \mid v \in e \}$.
Likewise, for $E' \subseteq E(G)$, we write $f(E)$ to denote $\{ f(e) \mid e \in E' \}$.
Note that if two CQs are homomorphic, then also their associated hypergraphs are homomorphic, while the 
converse is, in general, not true.

A fractional edge cover $\mathbf{x}$ of a set $W \subseteq V(H)$ for a hypergraph $H=(V(H), E(H))$ is a (not necessarily optimal) solution of the linear program:
\begin{equation*}
  \begin{array}{ll@{}l}
    \text{minimize} & \sum_e x_e & \\
    \text{subject to} & \sum_{e : v \in e} x_e \geq 1 \quad \mbox{}  & \mbox{for all } v \in W \\
                    & x_e \geq 0 & \mbox{for all} \ e \in E(H) 
  \end{array}
\end{equation*}
If no subset is specified the cover of all vertices is assumed. We
call the result of the objective function the \emph{total weight} of a
cover. The minimal total weight of a fractional edge cover of $H$ is
the \emph{fractional edge cover number} $\rho^*(H)$. For a CQ $q$,
$\rho^*(q)$ is the fractional cover number of the hypergraph
associated with $q$.

\section{Semantic Properties and Core Minimal Functions}
\label{sec:semprop}

As stated in the introduction, our main interest is the complexity of
the underlying question posed by a query. To this end we move away
from looking at specific queries that implement the question and
instead consider the whole equivalence class of queries with the
intended output. Instead of looking at the width of a query we now
want to study the \emph{semantic width} of a query, the width of the
most \emph{efficient} way to equivalently formulate the query.

\begin{definition}
  Let $\mathcal{Q}$ be the class of all conjunctive queries and
  $w\colon \mathcal{Q} \to \mathbb{R}^+$. We define the \emph{semantic variant of $w$} as $\semantic{w}(q) := \inf \{ w(q') \mid q' \simeq q \}$.
\end{definition}

The previous definition also illustrates one of the main issues of
semantic width, as there are infinitely many equivalent CQs, it is
inherently unclear how it can be computed. The rest of this section
provides a framework to determine these semantic variants
in a more practical manner.

Barcel\'{o} et al. investigate what they call the reformulation
problem for ghw in~\cite{barcelo2017semantic}. This is the problem
whether, given a CQ $q$, there exists an equivalent query with a ghw
less or equal to some specified threshold.  We generalize their
reformulation result to the following Lemma~\ref{lem:reform} for a
more general class of functions.

\begin{definition}
  let $\mathcal{Q}$ be the class of all conjunctive queries. We call a function $w\colon \mathcal{Q} \to \mathbb{R}^+$ \emph{\mononm} if it is invariant under isomorphisms and for any $q \in \mathcal{Q}$: $w(Core(q)) \leq w(q)$.
\end{definition}

\begin{samepage}
\begin{lemma}
  \label{lem:reform}
  Fix $k \geq 1$, and let $w$ be a \mononm function. For each conjunctive query $q$ the following are equivalent:
  \begin{enumerate}
  \item There exists a $q'$ equivalent to $q$ with $w(q') \leq k$.
  \item $w(Core(q)) \leq k$.
  \end{enumerate}
\end{lemma}
\end{samepage}
\begin{proof}
  The core of $q$ is always equivalent to $q$ and therefore the upward
  implication follows. For the downward implication $w(Core(q')) \leq w(q')$ by definition. If $q'$ is equivalent to
  $q$, then their cores must be isomorphic, thus $w(Core(q)) = w(Core(q')) \leq w(q') \leq k$.
\end{proof}

It is easy to see that $w$ being \mononm~is in fact also a necessary
condition for Lemma~\ref{lem:reform} above. It also follows that
for \mononm~functions, the minimal value among equivalent queries is
always found in the core. This leads us to the following convenient lemma.

\begin{lemma}
  \label{lem:semcore}
  A function $w$ is \mononm if and only if for all conjunctive queries $q \colon \semantic{w}(q) = w(Core(q)))$.
\end{lemma}
\begin{proof}
  The implication from left to right is immediate from Lemma~\ref{lem:reform}.
  For the other direction we note that if $q' \simeq q$, then $\semantic{w}(q') \leq w(q)$ by definition.
  Thus, from $Core(q) \simeq q$ we see $w(Core(q))=\semantic{w}(q) \leq w(q)$.
\end{proof}

\section{Edge Covers for Homomorphic Hypergraphs}
\label{sec:homgraph}

We prove that the fractional cover number is \mononm as a consequence
of the fact that fractional edge covers of hypergraphs are preserved
by homomorphisms. This fact is used again in the proof of our result for semantic fractional hypertree width. By showing that $\rho^*$ is \mononm we also determine $\semantic{\rho^*}$.

\begin{lemma}
  \label{lem:homomorph_cover}
  Let $f$ be a homomorphism from $G$ to $H$. 
  Given a fractional edge cover $\mathbf{x}$ of $G$, define $\mathbf{x'}$ s.t. 
  $$x'_h = \sum_{g \in \xp{h}} x_g \qquad h \in E(H).$$
    Then $\mathbf{x'}$ is a fractional edge cover of $f(V(G))$ with the same total weight as $\mathbf{x}$.
\end{lemma}
\begin{proof}

  We first show that $\mathbf{x'}$ is fractional edge cover.  
  To see this, choose an arbitrary  $w \in f(V(G))$. For every $v \in f^{-1}(w)$, we have that
  $\sum_{g \in I_v} x_g \geq 1$. For every $E \subseteq E(G)$, 
  $E \subseteq f^{-1}(f(E))$ and, therefore, we also have 
  $$\sum_{h \in f(E)} x'_h = \sum_{h \in f(E)} \sum_{g \in f^{-1}(h)}  x_g 
  = \sum_{g \in f^{-1}(f(E))}  x_g   \geq \sum_{g \in E} x_g.$$
  From this we conclude:
  $$\sum_{h \in I_w} x'_h \geq \sum_{h \in f(I_v)}  x'_h \geq \sum_{g \in I_v} x_g \geq 1$$

The leftmost inequality holds, because $f(I_v) \subseteq I_w$. 
The rightmost inequality holds, because we are assuming that $\mathbf{x}$ is a fractional edge cover of $G$.
We have thus shown that 
$\mathbf{x'}$ covers $w$. Since  $w \in f(V(G))$ was arbitrarily chosen, we conclude that 
$\mathbf{x'}$ is a fractional edge cover of $f(V(G))$.

To see that the total weights of both covers are the same, observe:
  $$ \sum_{h \in f(E(G))} x'_h = \sum_{h \in f(E(G))} \sum_{g \in \xp{h}} x_g = \sum_{g \in E(G)} x_g$$
The right equality follows from the fact that every edge of $G$ is present in exactly one set $\xp{h}$, i.e., 
for $E = E(G)$, we actually have $E = f^{-1}(f(E))$.
\end{proof}

\begin{lemma}
  \label{lem:corecover}
  The fractional edge cover number $\rho^*$ of a conjunctive query is \mononm. 
\end{lemma}
\begin{proof}
  Let $G$ be the hypergraph of $q$ and $H$ be the hypergraph of
  $Core(q)$.  Since there is a surjective homomorphism from $q$ to
  $Core(q)$, there exists a surjective homomorphism from $G$ to $H$. Then, by
  Lemma~\ref{lem:homomorph_cover}, for any fractional edge cover of
  $G$ there exists a cover of $H$ with equal weight.
\end{proof}

\begin{theorem} For all conjunctive queries $q$:
  \label{thm:semrho}
  $$\semantic{\rho^*}(q) = \rho^*(Core(q)).$$
\end{theorem}

\section{Core Minimal Notions of Widths}
\label{sec:swidth}

In this section we show that fractional hypertree width, adaptive width and submodular width are all \mononm. The proof is constructive by transforming tree decompositions of a query to tree decompositions of its core in a way that can only decrease the width of the decomposition. 

We follow Marx~\cite{marx2013tractable} in our definitions of 
various notions of widths. A tuple $(T, (B_u)_{u \in V(T)})$ is a
\emph{tree decomposition} of a hypergraph $H$ if $T$ is a tree, every
$B_u$ (the \emph{bags}) is a subset of $V(H)$, for every $e \in E(H)$
there is a node in the tree s.t. $e \subseteq B_u$, and for every
vertex $v \in V(H)$, $\{u \in V(T) \mid v \in B_u\}$ is connected in
$T$. For functions $f\colon 2^{V(H)} \to \mathbb{R}^+$, the \emph{$f$-width}
of a tree decomposition is $\sup\{f(B_u) \mid u \in V(T)\}$ and the
$f$-width of a hypergraph is the minimal $f$-width over all its tree
decompositions. Let $\mathcal{F}$ be a class of functions from subsets
of $V(H)$ to the non-negative reals, then the $\mathcal{F}$-width of
$H$ is $\sup \{f\mbox{-width}(H)\mid f\in \mathcal{F} \}$. All such
widths are implicitly extended to conjunctive queries by taking the
width of the associated hypergraph.

The following properties of functions $b\colon 2^{V(H)}\to \mathbb{R}^+$ are important: 
\begin{itemize}
\item $b$ is called \emph{submodular} if $b(X) + b(Y) \geq b(X \cap Y) + b(X \cup Y)$ holds
for every $X \subseteq V(H)$. 
\item $b$ is called \emph{modular} if $b(X) + b(Y) = b(X \cap Y) + b(X \cup Y)$ holds
for every $X \subseteq V(H)$. 
\item $b$ is called \emph{edge-dominated} if $b(e) \leq 1$ for every $e \in E(H)$. 
\item Finally, $b$ is \emph{monotone} if $X \subseteq Y$ implies $b(X) \leq b(Y)$.
\end{itemize}
For $X \subseteq V(H)$, let $\rho_H(X)$ be the
size of the smallest integral edge cover of $X$ by edges in $E(H)$ and
$\rho_H^*(X)$ the size of the smallest fractional edge cover of $X$ by edges in $E(H)$.
 We are now ready to define the specific widths that are being
investigated:
\begin{definition} For a hypergraph $H$:
  \hfill 
  \begin{description}[topsep=0pt]
  \item[Generalized hypertree width of $H$] \cite{adler2007hypertree,2009gottlob}:  
  $ghw(H) :=\rho_H$-width.
  \item[Fractional hypertree width of $H$] \cite{2014grohemarx}:
  $fhw(H) := \rho_H^*$-width.
  \item[Adaptive width of $H$] \cite{marx2011tractable}:
  $adw(H) := \mathcal{F}$-width$(H)$,
    where $\mathcal{F}$ is the set of all monotone, edge-dominated,
    modular functions $b$ on $2^{V(H)}$ with $b(\emptyset)=0$.
    (Equivalently, $\mathcal{F}$ can be defined as the set of all functions
    $b\colon 2^{V(H)}\to \mathbb{R}^+$ obtained 
    as $b(X) = \sum_{v \in X} f(x)$, where $f$ is a fractional independent set of $H$.)
  \item[Submodular width of $H$] \cite{marx2013tractable}:
  $subw(H) := \mathcal{F}$-width$(H)$,
    where $\mathcal{F}$ is the set of all monotone, edge-dominated,
    submodular functions 
    $b$ on $2^{V(H)}$ with $b(\emptyset)=0$.    
  \end{description}
\end{definition}

For $ghw$ the following result was already shown
in~\cite{barcelo2017semantic}. However, because it comes for free with
our proof that $fhw$ is \mononm, we include $ghw$ in the following
Lemma~\ref{lem:swidth} to illustrate how the proof applies to $ghw$.

\begin{lemma}
\label{lem:swidth}
The functions $ghw$, $fhw$, $adw$, and $subw$ are \mononm.
\end{lemma}
\begin{proof}
Let $q$ be a conjunctive query and $f$ an
endomorphism from $q$ to $Core(q)$. W.l.o.g., we may assume $f(v)=v$ for all $v \in f(q)$).  
This can be seen as follows: suppose that $f(v)=v$ does not hold for all $v \in f(q)$).  
Clearly, $f$ restricted to $Core(q)$ must be a variable renaming. Hence, there exists 
the inverse variable renaming $f^{-1} \colon Core(q) \rightarrow Core(q)$. Now set 
$f^* = f^{-1}(f(\cdot))$. Then $f^* \colon q \rightarrow Core(q)$ is the desired endomorphism
from $q$ to $Core(q)$ with $f^*(v)=v$ for all $v \in f^*(q)$).  
  
Let $H=(V(H), E(H))$ denote the hypergraph of $q$ and $H'=(V(H'), E(H'))$ the hypergraph of 
$Core(q) = f(q)$.
Furthermore, let $(T, (B_u)_{u\in V(T)})$ be a tree decomposition of $H$. 
Then we create $(T, (B'_u)_{u \in V(T)})$ with the same structure as the
original decomposition and $B'_u = B_u \cap V(H')$.  This gives a
tree decomposition of $H'$: for every edge $e \in E(H')$ with 
$e \subseteq B_u$, also $e \subseteq B_u \cap V(H')$ holds, because
$e \subseteq V(H')$.  Removing vertices completely from a
decomposition cannot violate the connectedness condition. Actually, some
bags $B'_u$ might become empty but this is not problematic: either we simply allow empty bags in the 
definition of the various notions of width; or we transform $(T, (B'_u)_{u \in V(T)})$ by
deleting all nodes $u$ with empty bag from $T$ and append every node with a non-empty bag as a (further) child of the nearest ancestor node with non-empty bag.

  \begin{description}[topsep=0pt]
  \item[fhw (and ghw):] We show that if $(T, (B_u)_{u\in V(T)})$ has $\rho_H^*$-width $k$,
    then $(T, (B'_u)_{u \in V(T)})$ has $\rho_{H'}^*$-width $\leq k$: By
    assumption, there is a fractional edge cover $\gamma_u$ of every
    set $B_u$ with weight $\leq k$.  
    By Lemma~\ref{lem:homomorph_cover}, there exists a cover $\gamma'_u$ of
    $f(B_u)$ with weight $\leq k$ and because $B'_u \subseteq f(B_u)$,
    $\gamma'_u$ also covers $B'_u$. 
    
    The proof for $ghw$ is analogous
    (the cover created in Lemma~\ref{lem:homomorph_cover} preserves
    integrality).
  \item[subw (and adw):]
    Let $\mathcal{F}$ and $\mathcal{F}'$ be the sets of monotone, edge-dominated, submodular functions on $V(H)$ and $V(H')$ respectively. We show that for every $b' \in \mathcal{F}'$
    there exists $b \in \mathcal{F}$, such that  $b'$-width$(H') \leq b$-width$(H)$:
    
    Consider an arbitrary monotone edge-dominated submodular function 
    $b' \colon 2^{V(H')}\to \mathbb{R}^+$ with $b'(\emptyset)=0$. This function can be extended to a monotone, edge-dominated, 
    submodular function $b\colon 2^{V(H)} \to \mathbb{R}^+$ on $V(H)$ by setting
    $b(X) = b'(X \cap V(H'))$ for every $X \subseteq V(H)$. 
    Now, assume $(T, (B_u)_{u \in V(T)})$ has $b$-width $k$, then
    $(T, (B'_u)_{u \in V(T)})$ has $b'$-width $\leq k$ because
    $b'(B'_u) = b'(B_u \cap V(H')) = b(B_u)$ for every $u \in
    V(T)$. Thus, the $b'$-width of $H'$ is less or equal the $b$-width
    of $H$. As submodular width considers the supremum over all
    permitted functions we see that $subw(H') \leq subw(H)$.

    For $adw$ observe that the definition of function $b$ and the line of argumentation above 
    still holds if we start off with a {\em modular} function $b'\colon 2^{V(H)} \to \mathbb{R}^+$.
    
  \end{description}
\end{proof}

\begin{theorem}
  \label{thm:swidth}
  For every conjunctive query $q$:
  \begin{itemize}[topsep=0pt]
  \item $\sghw(q) = ghw(Core(q))$
  \item $\sfhw(q) = fhw(Core(q))$
  \item $\sadw(q) = adw(Core(q))$
  \item $\ssubw(q) = subw(Core(q))$
  \end{itemize}
\end{theorem}

To see that the semantic variants of functions from
Theorems~\ref{thm:swidth} and~\ref{thm:semrho} bring non-trivial
improvements consider grids of atoms with the same relation name. Under the right circumstances --
some care is necessary regarding the output variables and the ordering
of variables in the atoms -- their cores are often vastly simpler structures, in some cases even a single atom.

\section{Conclusion \& Future Work}
\label{sec:conclusion}

In this work we introduce the concept of semantic width, a measure for
the complexity of the underlying semantics of a query. We extend a
result of Barcel\'{o} by showing that the problem of determining the
semantic width of a query, which seems inherently undecidable, can in
fact be reduced to determining the width of the core for various
common notions of width and for the fractional cover number.
Therefore, we can compute the various
presented semantic widths by first computing the core and then its
width. However, finding the core of a CQ is an NP-complete
problem~\cite{chandra1977optimal}.
Some properties of CQs are known to make the computation of
generalized and fractional hypertree decompositions
tractable~\cite{fischl2018general}. It warrants study if these
conditions could also be used to make computing the core
tractable. If so, the respective semantic width also becomes tractable.

A natural next step would be to extend known complexity results to the
semantic viewpoint. As an example, consider the classical result by
Atserias et al.~\cite{atserias2008size}. It says that for full
conjunctive queries (CQs where every variable of the body is also in
the head) if a class of CQs $\mathcal{Q}$ has bounded fractional cover
number, computing the answers is in polynomial time for $\mathcal{Q}$.
Using the semantic fractional edge cover number and Theorem~\ref{thm:semrho}, we can extend this result
to general CQs whose cores are efficiently computable as stated in the following corollary. 
A similar result, generalizing the bounds on the answer size to general CQs was shown
in~\cite{gottlob2012size}.

\begin{corollary}
  Let $\mathcal{Q}$ be a class of conjunctive queries whose
  core is computable in polynomial time. Then, if queries in
  $\mathcal{Q}$ have bounded $\semantic{\rho^*}$, queries in
  $\mathcal{Q}$ can be computed in polynomial time.
\end{corollary}

The transformation of tree decompositions used in the proof of
Lemma~\ref{lem:swidth} suggests that the problem of finding a
decomposition for $Core(q)$ is already included in the problem of
finding a decomposition for a CQ $q$. There may be further ways to
exploit this connection. We are particularly interested in the
possibility of using tree decompositions of a query to speed up
finding the core.

Preliminary results of ongoing work suggest that knowledge of the
semantic generalized hypertree width of a query can be used to solve
the query efficiently even without knowing the core. This opens up an
exciting area of applications and we aim to expand on this topic soon.

\bibliography{../all}
\end{document}